%% file: main.tex
\documentclass[runningheads]{llncs}
\usepackage[T1]{fontenc}
\usepackage{graphicx}
\input{packages}

\input{macros}

\begin{document}
\newif\ifsubmission

\title{The Latency Price of Threshold Cryptosystem \\ in Blockchains}
%

\author{
Zhuolun Xiang\inst{1} \and Sourav Das\inst{2} \and Zekun Li\inst{1}\and Zhoujun Ma\inst{1}\and Alexander Spiegelman\inst{1}}
\authorrunning{Z. Xiang et al.}
%
\institute{Aptos Labs \and
University of Illinois Urbana-Champaign}
%
\maketitle              

\setlist{nosep}

\input{abstract}

\input{section/introduction}
\input{section/prelim}

\input{section/warmup}

\input{section/tight}

\input{section/ramp}
\input{section/evaluation}
\ifsubmission
\else
\input{appendix/related}

\fi

\section*{Acknowledgments}
We would like to thank Alin Tomescu, Andrei Tonkikh, and Benny Pinkas for helpful discussions.

%
%
%
%

\bibliographystyle{splncs04}
\bibliography{references}

\ifsubmission
\else

\clearpage

\appendix
\input{appendix/prelim}

\input{appendix/tight}

\input{appendix/ramp}

\input{appendix/evaluation}
\input{appendix/discussion}

\fi

\end{document}

%% file: packages.tex
\usepackage{graphicx}
\graphicspath{ {images/} }
\usepackage[font=small,skip=0pt]{caption}
\usepackage[font=small,skip=0pt]{subcaption}

\usepackage{algorithm}
\usepackage[noend]{algpseudocode}
\usepackage{mdwmath,mathtools}
\usepackage{mdwtab}
\usepackage[utf8]{inputenc}
\usepackage[english]{babel}
\usepackage{breqn}
\usepackage{tabto}
\usepackage{stmaryrd}
\usepackage{hyperref}

\usepackage{amsmath,amsfonts}
\usepackage{paralist}
\usepackage[normalem]{ulem}
\usepackage[T1]{fontenc}
\usepackage{booktabs}
\usepackage{multirow}
\usepackage{threeparttable}
\usepackage{multicol}
\usepackage[symbol]{footmisc}
\usepackage{makecell}

\usepackage{color}
\usepackage{url}
\usepackage{bbm}
\usepackage{bm}
\usepackage{graphics}
\usepackage[inline]{enumitem}
\usepackage{tikz-cd}
\usepackage{pgfplots}
\usepgfplotslibrary{groupplots}

\usepackage{tabularx}
\usepackage{xparse}
\usepackage[capitalise]{cleveref}
\usepackage{aliascnt}

\NewDocumentCommand{\statcirc}{ O{#2} m }{%
    \begin{tikzpicture}
    \fill[#2] (0,0) circle (1.0ex); 
    \fill[#1] (0,0) -- (180:1ex) arc (180:0:1ex) -- cycle; 
    \end{tikzpicture}
}

\usepackage{pifont}

\algdef{SE}[EVENT]{Upon}{EndUpon}[1]{\textbf{upon}\ #1\ \algorithmicdo}{\algorithmicend\ \textbf{event}}%
\algtext*{EndUpon}

\algnewcommand{\LineComment}[1]{\State \(\triangleright\) #1}

\newcommand{\para}[1]{\vspace{1pt}\noindent\textbf{#1}}

\usepackage[framemethod=tikz]{mdframed} 
\usepackage{tablefootnote} 
\makeatletter 
\AfterEndEnvironment{mdframed}{%
 \tfn@tablefootnoteprintout%
 \gdef\tfn@fnt{0}%
}
\makeatother 

\usepackage{pgfplots}
\usetikzlibrary{patterns}
\pgfplotsset{compat=1.18}

\newtheorem{game}[theorem]{Game}

\usepackage{xspace}
\usepackage{comment}

%% file: macros.tex
\newcommand{\setup}{{\sf Setup}}
\newcommand{\sgen}{{\sf ShareGen}}
\newcommand{\eval}{{\sf Eval}}
\newcommand{\peval}{{\sf PEval}}
\newcommand{\pver}{{\sf PVer}}
\newcommand{\comb}{{\sf Comb}}

\newcommand{\share}{{\tt SHARE}}

\newcommand{\PREFIN}{{\tt PREFIN}}

\newcommand{\bsigma}{\bm{\sigma}}

\newcommand{\pk}{{\sf pk}}

\newcommand{\val}{{\sf val}}
\newcommand{\m}{\ensuremath{m}}

\newcommand{\negl}{{\sf negl}}

\newcommand{\sch}[1]{\llbracket{#1}\rrbracket}

\newcommand{\one}{\vspace{1mm}}

\newcommand{\N}{\mathbb{N}}

\newcommand{\cE} {\mathcal{E}}
\newcommand{\cC} {\mathcal{C}}
\newcommand{\cS} {\mathcal{S}}
\newcommand{\cX} {\mathcal{X}}
\newcommand{\cK} {\mathcal{K}}

\newcommand{\ver}{{\sf Verify}}

\newcommand{\adv}{\ensuremath{\mathcal A}}

\newcommand{\cH}{{\mathcal H}}
\newcommand{\cM}{{\mathcal M}}

\newcommand{\rbcma}{\ensuremath{{\sf RB\mhyphen CMA}_{\tc}^\adv}}
\newcommand{\rbcmap}{\rbcma(\lambda)}

\newcommand{\upcma}{\ensuremath{{\sf UP\mhyphen CMA}_{\tc}^\adv}}
\newcommand{\upcmap}{\upcma(\lambda)}

\newcommand{\pp}{pp}

\mathchardef\mhyphen="2D

\newcommand{\ths}{t_{\sf sec}}
\newcommand{\thr}{t_{\sf rec}}
\newcommand{\thf}{t_{\sf fin}}
\newcommand{\grt}{{\sf GRT}}
\newcommand{\gft}{{\sf GFT}}
\newcommand{\gftr}{\gft_{\rd}}
\newcommand{\grtr}{\grt_{\rd}}
\newcommand{\latency}{{\sf L}}

\newcommand{\longtcs}{$(n,\ths,\thr)$-threshold cryptosystem\xspace}
\newcommand{\posrand}{Das et al.\xspace}
\newcommand{\name}{\ensuremath{\sf BTC}\xspace}
\newcommand{\longname}{blockchain-native threshold cryptosystem\xspace}
\newcommand{\nameft}{\ensuremath{\sf BTC_{FT}}\xspace}
\newcommand{\longnameft}{blockchain-native threshold cryptosystem with finalization threshold\xspace}
\newcommand{\MBB}{Multi-shot Byzantine Broadcast\xspace}
\newcommand{\MBBFT}{\MBB with finalization threshold\xspace}
\newcommand{\round}{round\xspace}
\newcommand{\rounds}{rounds\xspace}
\newcommand{\rd}{\ensuremath{r}\xspace}

\newcommand{\tc}{\ensuremath{\sf TC}\xspace}

\newcommand{\mbb}{\ensuremath{\sf MBB}\xspace}
\newcommand{\mbbft}{\ensuremath{\sf MBB_{FT}}\xspace}

\newcommand{\aptos}{Aptos\xspace}
\newcommand{\jolteon}{Jolteon\xspace}

\newcommand{\bcast}{{\sf bcast}}

\newcommand{\fin}{{\sf finalize}}
\newcommand{\queue}{\ensuremath{queue}}

\newcommand{\slow}{slow path\xspace}
\newcommand{\fast}{fast path\xspace}
\newcommand{\fastrand}{fast-path\xspace}

%% file: abstract.tex
\begin{abstract}
Threshold cryptography is essential for many blockchain protocols. For example, many protocols rely on threshold common coin to implement asynchronous consensus, leader elections, and randomized applications. Similarly, threshold decryption and threshold time-lock puzzles are often necessary for privacy.

In this paper, we study the interplay between threshold cryptography and a class of blockchains that use Byzantine-fault tolerant (BFT) consensus protocols with a focus on latency. More specifically, we focus on \emph{blockchain-native threshold cryptosystem}, where the blockchain validators seek to run a threshold cryptographic protocol once for every block with the block contents as an input to the threshold cryptographic protocol.
All existing approaches for blockchain-native threshold cryptosystems introduce a latency overhead of at least one message delay for running the threshold cryptographic protocol. 
In this paper, we first propose a mechanism to eliminate this overhead for blockchain-native threshold cryptosystems with \emph{tight} thresholds, i.e., in threshold cryptographic protocols where the secrecy and reconstruction thresholds are the same. 
However, real-world proof-of-stake-based blockchain-native threshold cryptosystems rely on ramp thresholds, where reconstruction thresholds are strictly greater than secrecy thresholds. For these blockchains, we demonstrate that the additional delay is unavoidable. We then introduce a mechanism to minimize this delay in the optimistic case.
We implement our optimistic protocol for the proof-of-stake distributed randomness scheme on the Aptos blockchain. Our measurements from the Aptos mainnet show that the optimistic approach reduces latency overhead by 71\%, from 85.5 ms to 24.7 ms, compared to the existing method.  
\end{abstract}

%% file: section/introduction.tex
\section{Introduction}
\label{sec:intro}
Threshold cryptography plays a vital role in modern blockchains, where various applications rely on primitives such as distributed randomness and threshold decryption. In threshold cryptography, a secret is shared among a set of parties using a threshold secret sharing~\cite{shamir1979share,blakley1979safeguarding}, and parties seek to collaboratively evaluate a function of the shared secret and some public input without revealing the shared secret. For security, the function of the shared secret and the public information is revealed only if a threshold fraction of parties contribute to the function evaluation. 

In this paper, we study the interplay between threshold cryptography and a class of blockchains that use Byzantine-fault tolerant (BFT) consensus protocols with a focus on latency. More specifically, we focus on \emph{blockchain-native threshold cryptosystem}, where the blockchain validators seek to run a threshold cryptographic protocol \tc\ once for every block with the block's content as an input to \tc\ protocol. We focus on schemes where the secret is shared using the Shamir secret sharing scheme~\cite{shamir1979share}, and the threshold cryptographic protocol is non-interactive, i.e., parties send a single message during the threshold cryptography protocol.

One concrete example of a \longname\ is the recent distributed randomness protocol for proof-of-stake blockchains~\cite{das2024distributed}, that has been deployed in the Aptos blockchain~\cite{aptos}.
In~\cite{das2024distributed}, parties collaboratively compute a threshold verifiable random function~(VRF) to generate shared randomness for each block, using the cryptographic hash of the block as an input to the threshold VRF. 
Similarly, Kavousi et al.~\cite{kavousi2023blindperm} propose to use threshold decryption to mitigate Maximal Extractable Value (MEV) attacks by the block proposers. Specifically, in \cite{kavousi2023blindperm} blockchain validators first order a set of encrypted transactions using a consensus protocol. Next, upon ordering, block validators run a threshold decryption protocol to decrypt the finalized transactions and execute them.

One limitation of existing \longname\ is that parties participate in the \tc\ protocol only after the block is finalized. Hence, all existing protocol introduces at least one additional message delay before the output of the \tc\ protocol is available, for the parties to exchange their \tc\ shares. As a result, blockchains that seek to use the output of the \tc\ protocol to execute the finalized transactions also incur this additional latency. For blockchains~\cite{castro1999practical,androulaki2018hyperledger} with optimal consensus latency of three-message delay~\cite{martin2006fast,abraham2021good,kuznetsov2021revisiting}, the additional round of communication adds at least 33\% latency overhead, which is significant. 

This paper studies whether the additional delay is inherent to support threshold cryptography in BFT-based blockchains. More specifically, let \tc\ be a threshold cryptography scheme. Then, the \emph{secrecy threshold} of \tc\ is the upper bound on the number of \tc\ messages an adversary can learn without learning the output of the \tc\ protocol. Alternatively, the \emph{reconstruction threshold} is the number of \tc\ messages an honest party requires to be able to compute the \tc\ output.
Committee-based blockchains where the parties have equal weights (stakes), such as Dfinity~\cite{hanke2018dfinity}, can use \longname\ with the same secrecy threshold as the reconstruction threshold.
For a wide variety of these blockchains, we present a protocol in which the parties can compute the \tc\ output simultaneously with the block finalization time. More specifically, our protocol applies to all BFT consensus protocols in which a value is finalized if and only if a threshold number of parties \emph{prefinalize} the value. 

However, many proof-of-stake blockchains~\cite{eth,solana,aptos,sui} where parties have unequal stakes, will rely on threshold cryptography with \emph{ramp} thresholds~\cite{blakley1985security}, i.e., use threshold cryptographic protocols where the reconstruction threshold is strictly larger than the secrecy thresholds. The ramp nature of threshold cryptography in these protocols is because these protocols assign to each party an approximate number of shares proportional to their stake~\cite{tonkikh2024swiper}. This approximate assignment of a number of shares to each party introduces a gap between the secrecy and reconstruction threshold, as the assignment process may allocate more shares to the corrupt parties and fewer shares to honest ones. Somewhat surprisingly, we prove a lower bound result illustrating that for \longname with ramp thresholds, the extra latency incurred by existing protocols is inherent for a wide family of consensus protocols. 


To circumvent this impossibility result, we propose a mechanism to design \longname\ protocols with ramp thresholds that achieve small latency overhead under optimistic executions. To demonstrate the effectiveness, we implement our solution atop the distributed randomness protocol (based on threshold VRF) used in the \aptos blockchain and evaluate its performance with their prior protocol. Our evaluation with real-world deployment illustrates that our optimistic approach reduces latency overhead by 71\%.

In summary, we make the following contributions:
\begin{itemize}
    \item We propose a mechanism (\Cref{alg:tight}) to remove the latency overhead for \longname with {\em tight} secrecy and reconstruction thresholds. The result applies to committee-based blockchain systems where parties have {\em equal} weights.

    \item We prove an impossibility result (\Cref{thm:impos}) indicating that the latency overhead is inherent for \longname with {\em ramp} thresholds, and present a solution (\Cref{alg:ramp}) that can remove the latency overhead under optimistic scenarios. 
    The results apply to proof-of-stake blockchain systems where parties have {\em unequal} weights.
    
    \item We implement our solution of ramp thresholds for distributed randomness and present evaluation numbers from the \aptos mainnet deployment. The evaluation demonstrates that the solution significantly improves the randomness generation latency overhead by $71\%$, from $85.5$ to $24.7$ ms.
\end{itemize}

\ifsubmission
Due to space constraints, we defer the related work to the full version~\cite{xiang2024latency}. 
\fi

%% file: section/prelim.tex
\section{Preliminaries}
\label{sec:prelim}
\para{Notations.}
For any integer $a$, we use $[a]$ to denote the ordered set $\{1,2,\ldots,a\}$. 
For any set $S$, we use $|S|$ to denote the size of set $S$. 
We use $\lambda$ to denote the security parameter. A machine is probabilistic polynomial time~(PPT) if it is a probabilistic algorithm that runs in time polynomial in $\lambda$. We use $\negl(\lambda)$ to denote functions that are negligible in $\lambda$. 
We summarize the notations in~\Cref{tab:notations}.

\input{section/table}

\para{System Model.}
\label{sec:prelim:system}
We consider a set of $n$ parties labeled $1,2,...,n$, where each party executes as a state machine.
For brevity, we present the results for parties with {\em equal} weights, which can be easily extended to the case with {\em unequal} weights.
The parties communicate with each other by message passing, via pairwise connected communication channels that are authenticated and reliable.
We consider a \emph{static} adversary $\adv$ that can corrupt up to $t$ parties before the execution of the system. 
A corrupted party can behave arbitrarily, and a non-corrupted party behaves according to its state machine. 
We say that a non-corrupted party is \emph{honest}.
We use $\cC$ to denote the set of corrupted parties, and $\cH$ to denote the set of honest parties.
The network is assumed to be {\em partially synchronous}, where there exists a known message delay upper bound $\Delta$, and a global stabilization time (GST) after which all messages between honest parties are delivered within $\Delta$~\cite{dwork1988consensus}. 
The adversary can receive messages from any party instantaneously.

\input{section/prelim-consensus}
\input{section/prelim-crypto}

\section{Blockchain-Native Threshold Cryptosystem}
\label{sec:btc}

We now formally define the problem of \longname. 
In such a system, a secret is shared among the participants in the blockchain protocol and these parties seek to collaboratively run a threshold cryptographic protocol, after every block, using the shared secret and the block as input.

\begin{definition}[Blockchain-Native Threshold Cryptosystem]\label{def:bntc}
    Let \mbb be a \MBB protocol as in~\Cref{def:mbb}.
    Let $\tc=(\setup,\sgen,\eval,\peval,\allowbreak \pver,\comb,\ver)$ be a \longtcs as in~\Cref{def:crypto}. 
    A \longname protocol $\name=(\mbb,\tc)$ is defined as follows. 
    \begin{enumerate}
        \item The parties start with a secret share of a secret key $s$ as per  $\sgen(s)$.
        \item The parties run \mbb, and may simultaneously execute the \tc\ protocol.
        \item Upon \mbb outputs $\fin(\rd,\m)$ for any \round $\rd\in\N$, after some possible delay, parties finish the \tc\ protocol to compute $\sigma=\eval(s,(\rd,\m))$ and outputs $(\rd, \m, \sigma)$.
    \end{enumerate}

    \noindent We require $\name$ to satisfy the following except for negligible probabilities. 
    \begin{itemize}
        \item \emph{Agreement}. For any \round $\rd\in\N$, if an honest party outputs $(\rd,\m,\sigma)$ and another honest party outputs $(\rd,\m',\sigma')$, then $\m=\m'$ and $\sigma=\sigma'$.
        \item \emph{Termination}. After GST, for any \round $\rd\in\N$ each honest party $i$ eventually outputs $(\rd,\m_i,\sigma)$ where $\m_i\in\cM\cup\{\bot\}$.
        \item \emph{Validity}. For any \round $\rd\in\N$, if the designated broadcaster $B_{\rd}$ is honest and calls $\bcast(\rd,\m)$ for $\m\in\cM$ after GST, then all honest parties eventually output $(\rd,\m,\eval(s,(\rd,\m)))$.
        \item \emph{Total Order}. If an honest party outputs $(\rd,\m,\sigma)$ before $(\rd',\m',\sigma')$, then $\rd<\rd'$.
        \item \emph{Secrecy}. If an honest party outputs $(\rd,\m,\sigma)$ where $\m\in\cM\cup\{\bot\}$, then $\sigma=\eval(s,(\rd,\m))$, and the adversary cannot compute $\eval(s,(\rd,\m'))$ where $\m\in\cM$ and $\m'\ne \m$.
    \end{itemize}
\end{definition}

For $\name=(\mbb,\tc)$, we use $\grtr$ to denote the global reconstruction time of \round \rd in \tc, as defined in~\Cref{def:grt}.

\paragraph{Example of \longname.}
On-chain distributed randomness generates a shared randomness for every finalized block.
\tc\ for this application can be a threshold VRF scheme. Upon \mbb (the blockchain consensus layer) outputs $\m$ (a block) for a round $r$, parties run the \tc\ protocol to compute the shared randomness $\eval(s,(r,m))$.

Below we define the latency of a \longname to measure the introduced latency overhead. 
Intuitively, $\latency_\rd$ is the maximum time difference, across all honest parties, between the time a honest party $i$ finalizes in \mbb\  and the time the same honest party $i$ outputs.
Since the transaction execution relies on the \tc output of the threshold cryptosystem, by definition, a party may have to wait a period of $\latency_{\rd}$ before executing the transactions finalized for \round \rd, thus increasing the blockchain's transaction end-to-end latency. 

\begin{definition}[Latency of Blockchain-Native Threshold Cryptosystem]\label{def:latency}
    During an execution of a \longname $\name=(\mbb,\tc)$, for any \round \rd and party $i$, let $T_{i,\rd}^{\sf F}$ be the physical time when party $i$ outputs $\fin(\rd, \m)$ for some $\m$ in $\mbb$, and $T_{i,\rd}^{\sf O}$ be the physical time when party $i$ outputs $(\rd,\m,\sigma)$ for some $\m,\sigma$ in $\name$. 
    The latency for \round \rd of the execution is defined to be $\latency_{\rd}=\max_{i\in\cH}(T_{i,\rd}^{\sf O}-T_{i,\rd}^{\sf F})$.
\end{definition}

\subsection{Blockchain-Native Threshold Cryptosystem with Finalization Threshold}
\label{sec:btc}

This paper focuses on a family of \longname protocols defined as follows.

\begin{definition}[Blockchain-Native Threshold Cryptosystem with Finalization Threshold]\label{def:btc}
    A \longnameft, or $\nameft=(\mbbft, \tc)$, is a \longname (\Cref{def:bntc}) that uses an \mbbft protocol (\Cref{def:ft}).
\end{definition}

We will henceforth shorten \emph{\longname} as \name, and \emph{\longnameft} as \nameft.

Some of the paper's results hold under optimistic conditions defined below. 

\para{Error-free.} An execution of \name is error-free if all parties are honest.

\para{Synchronized execution.} An execution of $\nameft=(\mbbft, \tc)$ is synchronized for a \round \rd if all honest parties prefinalize the same message $\m\in\cM$ for \rd at the same physical time~\footnote{This condition is defined solely for proving theoretical latency claims. 
In practice, different honest parties may prefinalize at different physical times, and the result of the paper still achieves latency improvements as in~\Cref{sec:rand:eval}.} in \mbbft.

\para{Optimistic.} An execution of \nameft is optimistic if the execution is error-free and synchronized for any \round \rd, and all messages have the same delay. 




%% file: section/table.tex
\begin{table}[t!]
    \centering
    \small
    \setlength\tabcolsep{8pt}
    \renewcommand{\arraystretch}{1.3}
    \begin{tabular}{cl}
        \toprule
        \textbf{Symbol} & \textbf{Description} \\
        \midrule
        \mbb & multi-shot Byzantine broadcast (\Cref{def:mbb}) \\
        \mbbft & \mbb with finalization threshold (\Cref{def:ft}) \\
        \tc & threshold cryptosystem (\Cref{def:crypto}) \\
        \name & \longname (\Cref{def:bntc}) \\
        \nameft & \name with finalization threshold (\Cref{def:btc}) \\
        $\ths$ & secrecy threshold in \tc (\Cref{def:crypto}) \\
        $\thr$ & reconstruction threshold in \tc (\Cref{def:crypto}) \\
        $\thf$ & finalization threshold in \mbbft (\Cref{def:ft}) \\ 
        $\gftr$ & global finalization time of \round \rd in \mbb (\Cref{def:gft}) \\
        $\grt$ & global reconstruction time in \tc (\Cref{def:grt}) \\
        $\grtr$ & global reconstruction time of \round \rd in \name \\
        $\latency_{\rd}$ & latency of \round \rd in \name (\Cref{def:latency})\\
        \bottomrule
    \end{tabular}
    \vspace{1mm}
    \caption{Table of Notations. 
    }
    \vspace{-7mm}
    \label{tab:notations}
\end{table}

%% file: section/prelim-consensus.tex
\subsection{Blockchain Definitions}
\label{sec:prelim:consensus}

We define \MBB under partial synchrony as follows to capture the consensus layer of many real-world blockchains that assume a partial synchronous network. 
We will use \MBB and consensus interchangeably throughout the paper. 
Intuitively, \MBB consists of infinite instances of Byzantine Broadcast with rotating broadcasters and guarantees a total ordering among all instances. 
The primary reason for introducing a new definition, rather than relying on existing ones such as Byzantine Atomic Broadcast, is the necessity of incorporating {\em rounds} into the definition, as rounds will be referenced in later definitions (\Cref{def:ft}, \Cref{def:bntc}).
The definition of \MBB captures many existing partially synchronous leader-based BFT protocols, or with minor modifications~\footnote{Many chained BFT protocols such as HotStuff~\cite{yin2019hotstuff} and Jolteon~\cite{gelashvili2022jolteon} achieve a weaker Validity property. In these protocols, the finalization of the message proposed by the broadcaster of \round \rd requires multiple consecutive honest broadcasters starting from \round \rd. This weaker Validity does not affect the results we present in this paper.}, such as~\cite{castro1999practical,buchman2016tendermint,gueta2019sbft,yin2019hotstuff,chan2020streamlet,jalalzai2020fast,gelashvili2022jolteon,doidge2024moonshot}, as well as DAG-based BFT protocols, such as~\cite{spiegelman2022bullshark,spiegelman2023shoal,arun2024shoal++,keidar2022cordial,babel2023mysticeti}. 


\begin{definition}[\MBB]\label{def:mbb}
    \MBB is defined for a message space $\cM$ where $\bot\notin\cM$, and \rounds $\rd=0,1,2,...$ where each \round $\rd\in\N$ has one designated broadcaster $B_{\rd}$ who can call $\bcast(\rd,\m)$ to broadcast a message $\m\in\cM$. 
    For any round $\rd\in\N$, each party can output $\fin(\rd,\m)$ once to finalize a message $\m\in\cM\cup\{\bot\}$.
    The \MBB problem satisfies the following properties. 
    \begin{itemize}
        \item \emph{Agreement}. For any \round $\rd\in\N$, if an honest party $i$ outputs $\fin(\rd,\m)$ and an honest party $j$ outputs $\fin(\rd,\m')$, then $\m=\m'$. 
        \item \emph{Termination}. After GST, for any \round $\rd\in\N$ each honest party $i$ eventually outputs $\fin(\rd,\m_i)$ where $\m_i\in\cM\cup\{\bot\}$.
        \item Validity. If the broadcaster $B_{\rd}$ of \round $\rd$ is honest and calls $\bcast(\rd,\m)$ for any $m\in \cM$ after GST, then all honest parties eventually output $\fin(\rd,\m)$.
        \item \emph{Total Order}. If an honest party outputs $\fin(\rd,\m)$ before $\fin(\rd',\m')$, then $\rd<\rd'$.
    \end{itemize}
\end{definition}

A \MBB protocol \mbb defines the state machine for each party to solve \MBB. 
Now we define an {\em execution} of an \mbb protocol, and the {\em globally finalization} of a message for a round. 

\begin{definition}[Execution, multivalent and univalent state]
    A configuration of the system consists of the state of each party, together with all the messages in transit. 
    Each execution of a \MBB protocol is uniquely identified by the sequence of configurations. 
    
    During the execution of a \MBB protocol, 
    the system is in a multivalent state for \round \rd, if there exist two possible executions $\cE\ne\cE'$ both extending the current configuration, where some honest party output differently in $\cE,\cE'$; 
    the system is in a univalent state of $\m\in\cM\cup\{\bot\}$ for \round \rd, if for all executions extending the current configuration, all honest parties always outputs $\fin(\rd,\m)$.
\end{definition}

\begin{definition}[Global finalization]\label{def:gft}
    During the execution of a \MBB protocol, a message $\m\in\cM\cup\{\bot\}$ is globally finalized for \round \rd, if and only if the system is in the univalent state of $\m$ for \rd. 
    The global finalization time $\gftr$ of \round \rd is defined as the earliest physical time when a message is globally finalized for \rd.

    We say that a party locally finalizes $\m$ for \round $\rd$ when it outputs $\fin(\rd,\m)$.
\end{definition}

Intuitively, a message is globally finalized for a \round \rd in \MBB when it is the only message that can be the output of \rd.
Global finalization is a global event that may not be immediately known to any honest party, but must occur no later than the moment that any honest party outputs $\fin(\rd,\cdot)$. 
Compared to local finalization, which occurs when any honest party outputs $\fin(\rd,\cdot)$, global finalization is more fundamental.

\para{\MBB with Finalization Threshold.}
The paper focuses on a family of \mbb protocols that have a finalization threshold $\thf$ defined as follows, where $\thf$ is a parameter of the definition.
We call such a protocol \MBBFT, or \mbbft.

\begin{definition}[Finalization Threshold $\thf$]\label{def:ft}
    For any \round $\rd\in\N$, the \MBBFT, or \mbbft, has a step where a party prefinalizes $\m$ for round $\rd$ by sending $(\PREFIN,\rd,\m)$ to all parties, such that
    \begin{itemize}
        \item\label{def:ft:1} For any \round \rd, there exists at most one $\m\in\cM$ such that any honest party may send $(\PREFIN,\rd,\m)$.
        \item\label{def:ft:2} For any \round \rd, any honest party can send $(\PREFIN,\rd,\bot)$ after sending $(\PREFIN,\rd,\m)$ for some $\m\in\cM$, but not the reverse.
        \item\label{def:ft:3} $\m\in\cM\cup\{\bot\}$ is globally finalized for \rd, if and only if there exist $\thf$ parties (or equivalently $\thf-|\cC|$ honest parties) that have sent $(\PREFIN,\rd,\m)$.
    \end{itemize}

    Any honest party outputs $\fin(\rd,\m)$ to locally finalize a message $\m\in\cM\cup\{\bot\}$ for a \round $\rd$ when the party receives $(\PREFIN,\rd,\m)$ messages from $\thf$ parties, which implies that $\m$ is globally finalized for \rd.

    We say $\thf$ is the finalization threshold of \mbbft. 
\end{definition}

Intuitively, prefinalization is a local state of any party. When enough honest parties prefinalize a message, the message is globally finalized. A single party prefinalizing a message does not guarantee that the message will be finalized, and the party may finalize another message at the end. 
In many BFT protocols such as HotStuff~\cite{yin2019hotstuff} and Jolteon~\cite{gelashvili2022jolteon}, prefinalization is also named \emph{lock}. In PBFT~\cite{castro1999practical}, party prefinalizes a message when sending a {\em commit} for the message.

\paragraph{Examples of \MBBFT.}
A larger number of \mbb protocols used in partially synchronous blockchains fall into this family. 
As a concrete example, \jolteon~\cite{gelashvili2022jolteon} is a partially synchronous \mbbft protocol deployed by blockchains such as \aptos~\cite{aptos} and Flow~\cite{flow}.
We explain in detail how \jolteon satisfies the definition of \MBBFT with finalization threshold
\ifsubmission
in the full version~\cite{xiang2024latency}.
\else
in~\Cref{app:rand:impl}.
\fi
Other than \jolteon,
numerous partially synchronous BFT protocols are also part of this family or can be easily adapted to fit into this family, such as PBFT~\cite{castro1999practical}, Tendermint~\cite{buchman2016tendermint}, SBFT~\cite{gueta2019sbft}, HotStuff~\cite{yin2019hotstuff}, Fast-Hotstuff~\cite{jalalzai2020fast}, Moonshot~\cite{doidge2024moonshot} and many others.
Additionally, another series of DAG-based consensus protocols also satisfy the definition of \MBBFT with finalization threshold, such as Shoal++~\cite{arun2024shoal++}, Sailfish~\cite{shrestha2024sailfish}, Cordial miners~\cite{keidar2022cordial} and Mysticeti~\cite{babel2023mysticeti}.  

%% file: section/prelim-crypto.tex
\subsection{Cryptography Definitions}
\label{sec:prelim:crypto}
Next, we describe the syntax and security definitions for threshold cryptosystems. We focus on non-interactive threshold cryptographic protocols.
\begin{definition}[Threshold Cryptosystem]
\label{def:crypto}
Let $\ths,\thr,n$ with $\ths\le \thr\le n$ be natural numbers. We refer to $\ths$ and $\thr$ as the secrecy and reconstruction threshold. Let $\cX$ be a input space. Looking ahead, the input space $\cX$ denotes the output space of the underlying \MBB protocol.

A $(n,\ths,\thr)$-threshold cryptosystem \tc\ is a tuple of PPT algorithms $\tc=(\setup,\sgen,\eval,\peval,\allowbreak \pver,\comb,\ver)$ defined as follows:
\begin{enumerate}
    \item $\setup(1^\lambda)\rightarrow \pp.$ The setup algorithm takes as input a security parameter and outputs public parameters $\pp$ (which are given implicitly as input to all other algorithms).
    
    \item $\sgen(s)\rightarrow \{\pk,\pk_i,\sch{s}_i\}_{i\in [n]}.$ The share generation algorithm takes as input a secret $s\in \cK$ from a secret key space $\cK$ and outputs a public key $\pk$, a vector  of threshold public keys $\{\pk_1,\ldots,\pk_n\}$, and a vector of secret shares $(\sch{s}_1,\ldots,\sch{s}_n)$. The $j$-th party receives $(\{\pk_i\}_{i\in [n]},\sch{s}_j)$.

    \item $\eval(s,\val)\rightarrow \sigma.$ The evaluation algorithm takes as input a secret share $s$, and a value $\val\in \cX$. It outputs a function output $\sigma$, which is called the \tc output in the paper.
    
    \one
    \item $\peval(\sch{s}_i,\val)\rightarrow \sigma_i.$ The partial evaluation takes as input a secret share $\sch{s}_i$, and a value $\val\in \cX$. It outputs a function output share $\sigma_i$, which is called the \tc output share in the paper.
    \one
    \item $\pver(\pk_i,\val,\sigma_i)\rightarrow 0/1.$ The partial verification algorithm takes as input a public key $\pk_i$, a value $\val$, and a \tc output share $\sigma_i$. It outputs 1 (accept) or 0 (reject).
    \one
    \item $\comb(S,\val,\{(\pk_i,\sigma_i)\}_{i\in S})\rightarrow \sigma/\bot.$ The combine algorithm takes as input a set $S\subseteq [n]$ with $|S|\ge \thr$, a value $\val$, and a set of tuples $(\pk_i,\sigma_i)$ of public keys and \tc output shares of parties in $S$. It outputs a \tc output $\sigma$ or $\bot$.

    \one 
    \item $\ver(\pk, \val,\sigma)\rightarrow 0/1:$ The verification algorithm takes as input a public key $\pk$, input $\val$, and evaluation output $\sigma$. It outputs 1 (accept) or 0 (reject).
\end{enumerate}
\end{definition}


We require a threshold cryptosystem to satisfy the standard \emph{Robustness} and \emph{Secrecy} properties, as defined 
\ifsubmission
in the full version~\cite{xiang2024latency}
\else
in~\Cref{app:prelim:crypto}
\fi
due to space constraints.

\begin{definition}[Ramp~\cite{blakley1985security} and tight thresholds]\label{def:crypto:thresholds}
For any $(n,\ths,\thr)$-threshold cryptosystem, we call it a \emph{tight} threshold cryptosystem if $\ths=\thr$, and a \emph{ramp} threshold cryptosystem if $\ths < \thr$.    
\end{definition}


\begin{definition}[Global reconstruction]
\label{def:grt}
For any \longtcs with secret $s$ and input $\val$, we say $\eval(s,\val)$ is \emph{globally reconstructed} if and only if the adversary learns $\eval(s,\val)$ (or equivalently $\thr-|\cC|$ honest parties reveal \tc output shares to \adv). The global reconstruction time $\grt$ is defined to be the earliest physical time when $\eval(s,\val)$ is globally reconstructed (i.e., when $\thr-|\cC|$ honest parties reveal \tc output shares to \adv).   

We say that a party locally reconstructs $\eval(s,\val)$ when it learns $\eval(s,\val)$ (or equivalently receiving $\thr$ valid shares).
\end{definition}

\para{Double sharing of the secret.}
Looking ahead, we require our threshold cryptosystem to support double sharing of the same secret for two sets of thresholds $(\ths,\thr)$ and $(\ths',\thr')$ where $(\ths,\thr)\ne (\ths',\thr')$. Threshold cryptosystems based on Shamir secret sharing~\cite{shamir1979share} easily support double sharing. 


%% file: section/warmup.tex
\subsection{A Strawman Protocol}
\label{sec:ramp:strawman}
\input{algos/warmup}
%
As a warm-up, we first describe a strawman protocol for any \longname\ $\name=(\mbb,\tc)$ (\Cref{def:bntc}) in~\Cref{alg:strawman}, which works for both tight and ramp thresholds. 
The protocol has a latency $\latency_{\rd}\geq \delta$ for any \round $\rd\in \N$ even in error-free executions with constant message delay $\delta$ between honest parties. 
To the best of our knowledge, all existing \longname follow this approach, such as the Dfinity~\cite{hanke2018dfinity} and Sui~\cite{sui} blockchains for distributed randomness.
%
For brevity, we refer to this protocol as the \emph{\slow}. 

As part of the setup phase, each party $i$ receives $(\{\pk\}_{i\in [n]}, \sch{s}_i)$, where $\sch{s}_i$ is the secret share of party $i$ and $\{\pk\}_{i\in [n]}$ is the vector of threshold public keys of all parties. 
Each party maintains a First-in-first-out (FIFO) \queue\ to record the finalized rounds awaiting the \tc output. These rounds are pushed into the FIFO \queue\ in the order they are finalized, and only the head of the FIFO \queue\ will pop and be output. Looking ahead, this ensures Total Ordering, even when parties reconstruct \tc\ outputs of different rounds in out of order. 
Each party additionally maintains two maps $\bm{\m}$ and $\bm{\sigma}$ to store the finalized message and \tc output shares of each parties  of each round, respectively.

In the protocol, for any given \round\ \rd, each party $i$ waits until a message $\m$ is finalized by the \mbb\ protocol in \round\ \rd. Upon finalization, each party $i$ computes its \tc\ output share $\sigma_i\gets\peval(\sch{s}_i,(\rd,\m))$ and sends the \share\ message $(\share,\rd,\m,\sigma_i)$ to all parties.
Party $i$ also adds \round \rd to $\queue$ and updates $\bm{\m}$ as $\bm{\m}[\rd]\gets m$.
Next, upon receiving $(\share,\rd,\m,\sigma_j)$ from party $j$, party $i$ first validates $\sigma_j$ using $\pver$ algorithm and adds $\sigma_j$ to the set $S_{\rd,\m}$ upon successful validation. Finally, upon receiving $\thr$ valid \share\ messages for $(\m,\rd)$, party $i$ computes the \tc\ output $\sigma$ using $\comb$ algorithm and updates $\bsigma$ as $\bsigma[\rd]=\sigma$.
%
%
Whenever party $i$ has the \tc output of \round \rd that is the head of $\queue$, party $i$ pops the queue and outputs the result $(\rd, \bm{\m}[\rd],\bm{\sigma}[\rd])$ for \round \rd. 

To ensure the Termination property, the reconstruction threshold must be no greater than $n-t$, i.e., $\thr\leq n-t$. Intuitively, this ensures that once the \mbb\ outputs in a round, every honest party receives a sufficient number of \tc\ output shares to reconstruct the \tc\ output. Additionally, for Secrecy for the strawman protocol, the secrecy threshold must be greater than the number of \tc\ shares controlled by the adversary, i.e.,  $\ths\geq t+1$. Intuitively, this prevents the adversary from reconstructing the \tc output on its own. 
%
The correctness of the protocol is straightforward and is omitted here for brevity.

We will now argue that the \slow has a latency overhead of at least $\delta$ even in error-free executions with constant message delay $\delta$ between honest parties.
For any \round \rd, any party needs to receive at least $\thr-|\cC|\geq 1$ shares from the honest parties to compute $\sigma$.
Consider the first honest party $i$ that outputs $\fin(\rd,\m)$. In the strawman protocol, party $i$ needs to wait for at least one message delay starting from finalization to receive the shares from the honest parties to compute $\sigma$, since other honest parties only send shares after finalization. 
Adding one additional message delay to the system represents a significant overhead as the \mbb latency can be as short as three message delays

%% file: algos/warmup.tex
\begin{algorithm}[t!]
  \caption{Slow Path for Blockchain-Native Threshold Cryptosystem}
  \label{alg:strawman}
  \small
  \begin{minipage}[t]{\columnwidth}
    \medskip
    SETUP:
    \begin{algorithmic}[1]
        \State let $\bm{\m}\gets\{\}$, $\bm{\sigma}\gets \{\}$
        \Comment{Maps that store outputs for rounds}
        \State let $\queue\gets\{\}$
        \Comment{A FIFO queue that stores the finalized rounds}
        \State let $(\{\pk_j\}_{j\in [n]},\sch{s}_i)\gets\sgen(s)$ for thresholds $t+1\leq \ths\leq \thr\leq n-t$
        \Comment{$s$ is unknown to any party}
    \end{algorithmic}
    
    \medskip\hrule\medskip
    SLOW PATH: 
    \begin{algorithmic}[1]
        \Upon{$\fin(\rd,\m)$}
            \State let $\bm{\m}[\rd]\gets\m$ and $\sigma_i\gets \peval(\sch{s}_i,(\rd,\m)) $
            \State $\queue.push(\rd)$
            \State {\bf send} $(\share,\rd,\m,\sigma_i)$ to all parties
        \EndUpon
    \end{algorithmic}

    \medskip\hrule\medskip
    RECONSTRUCTION:
    \begin{algorithmic}[1]
        \Upon{receiving $(\share,\rd,\m,\sigma_j)$ from party $j$}
            \If{$\pver(\pk_j,(\rd,\m),\sigma_j) = 1$}
                \State $S_{\rd,\m}\gets S_{\rd,\m}\cup \{j\}$
                \If{$|S_{\rd,\m}|\ge \thr$ and $\bm\sigma[\rd]=\{\}$}
                    \State let $\bm\sigma[\rd]\gets \comb(S_{\rd,\m},(\rd,\m),\{(\pk_i,\sigma_i)\}_{i\in S_{\rd,\m}})$
                \EndIf
            \EndIf
        \EndUpon
    \end{algorithmic}
    
    \medskip\hrule\medskip
    OUTPUT:
    \begin{algorithmic}[1]
        \Upon{$\bm\sigma[\queue.top()]\ne\{\}$}
        \Comment{Always running in the background}
            \State let $\rd\gets\queue.pop()$
            \State {\bf output} $(\rd,\bm{\m}[\rd],\bm\sigma[\rd])$
        \EndUpon
    \end{algorithmic}
  \end{minipage}
\end{algorithm}

%% file: section/tight.tex
\section{Tight Blockchain-Native Threshold Cryptosystem}
\label{sec:tight}
\input{algos/tight}
In this section, we present a protocol for $\nameft=(\mbbft,\tc)$ (\Cref{def:btc}) for tight thresholds that has low latency.
In any \round $\rd\in\N$ of any execution, the global finalization time of our protocol is same as the global reconstruction time, i.e., $\gftr=\grtr$. Moreover, in error-free executions~\footnote{We require error-free for the $\latency_{\rd}=0$ claim, otherwise malicious parties may cause honest parties to prefinalize at different times and lead to $\latency_{\rd}>0$.}, honest parties in our protocol learns the \tc\ output simultaneously with the \mbbft\ output,
i.e., $\latency_{\rd}=0$. We summarize our construction in~\Cref{alg:tight} and describe it next.

The setup phase is identical to that of~\Cref{alg:strawman}, except that the secrecy and reconstruction thresholds are set to be equal to the finalization threshold of \mbbft. 
Note that, the Termination property of \mbbft\ requires $\thf\leq n-t$, as honest parties needs to finalize a message even when corrupted parties do not send any throughout the protocol. 
Next, unlike~\Cref{alg:strawman}, parties reveal their \tc output shares when they prefinalize a message. 
More specifically, for every \round\ \rd, each party $i$ computes $\sigma_i\gets\peval(\sch{s}_i,(\rd,\m))$ upon prefinalizing the value $(\rd,\m)$, and sends the \share\ message $(\share,\rd,\m,\sigma_i)$ to all parties in addition to sending the \PREFIN\ message $(\PREFIN,\rd,\m)$. 
When a party receives $\thf$ \PREFIN\ messages $(\PREFIN,\rd,\m)$, it finalizes the message $\m$ for \round \rd, by adding \round \rd to $\queue$ and recording $\m$ in $\bm{\m}[\rd]$. The party also computes and sends $\sigma_i$ if it has not done so. 
The reconstruction and output phases are also identical to~\Cref{alg:strawman}, where parties collect and combine shares to generate \tc output, and output the result round-by-round.
We defer the protocol analysis 
\ifsubmission
to the full version~\cite{xiang2024latency}
\else
to~\Cref{app:tight}
\fi
due to space constraints.

%% file: algos/tight.tex
\begin{algorithm}[t!]
  \caption{Tight Blockchain-Native Threshold Cryptosystem}
  \label{alg:tight}
  \small
  \begin{minipage}[t]{\columnwidth}
    \medskip
    SETUP
    is same as~\Cref{alg:strawman} except that $\ths=\thr=\thf$.
    
    \medskip\hrule\medskip
    PREFINALIZATION:
    \begin{algorithmic}[1]
        \Upon{sending $(\PREFIN,\rd,\m)$ to all parties}
            \Comment{Augment the consensus protocol}
            \State let $\sigma_i\gets \peval(\sch{s}_i,(\rd,\m)) $
            \State {\bf send} $(\share,\rd,\m,\sigma_i)$ to all parties
        \EndUpon

        
        \Upon{$\fin(\rd,\m)$}
            \State let $\bm{\m}[\rd]\gets\m$
            \State $\queue.push(\rd)$
            \If{$(\share,\rd,*,*)$ not sent}
                \State let $\sigma_i\gets \peval(\sch{s}_i,(\rd,\m)) $
                \State {\bf send} $(\share,\rd,\m,\sigma_i)$ to each other party
            \EndIf
        \EndUpon
    \end{algorithmic}


    
    \medskip\hrule\medskip
    RECONSTRUCTION and OUTPUT are same as~\Cref{alg:strawman}.
  \end{minipage}
\end{algorithm}

%% file: section/ramp.tex
\section{Ramp Blockchain-Native Threshold Cryptography}
\label{sec:ramp}

In this section, we present an impossibility result and a feasibility result for $\nameft=(\mbbft,\tc)$ (\Cref{def:btc}) with ramp thresholds $\ths<\thr$.
%

\subsection{Impossibility}
\label{sec:ramp:impos}

First, we demonstrate the impossibility result, which says that no $\nameft$ protocol with ramp thresholds $\ths<\thr$ can always guarantee that global finalization and reconstruction occur simultaneously. 


\begin{theorem}\label{thm:impos}
    For any $\nameft=(\mbbft,\tc)$ with ramp thresholds, there always exists some execution where $\grtr>\gftr$ for each \round $\rd\in\N$.
\end{theorem}

Due to space constraints, the proof of~\Cref{thm:impos} is deferred 
\ifsubmission
to the full version~\cite{xiang2024latency}.
\else
to~\Cref{app:ramp:impos}.
\fi
\Cref{thm:impos} states that for any \longnameft and ramp thresholds, there always exists an execution where global reconstruction occurs after global finalization for each \round. 
In fact, existing solutions for \name with ramp thresholds all have a latency of at least one message delay, such as \posrand~\cite{das2024distributed}.

\subsection{Fast Path}
\label{sec:ramp:opt}

\input{algos/ramp}

\Cref{thm:impos} claims that no \nameft protocol with ramp thresholds can {\em always} guarantee that the global finalization and reconstruction occur simultaneously, implying that the latency of \nameft may be unavoidable.  
Fortunately, we can circumvent this impossibility result in optimistic executions.
In this section, we describe a simple protocol named {\em \fast} that, for any \round \rd, achieves $\grtr=\gftr$ under synchronized executions, and $\latency_{\rd}=0$ under optimistic executions~\footnote{Similar to~\Cref{sec:tight}, we require error-free for the $\latency_{\rd}=0$ claim. However, the honest parties can reconstruct the \tc output via fast path as long as $|\cH|\geq \thr'$ (with $\latency_{\rd}>0$).}. As we illustrate in \Cref{sec:rand:eval}, in practice, our new protocol achieves significantly lower latency compared to the strawman protocol. 

The key observation from~\Cref{thm:impos} is that, to ensure the same global finalization and reconstruction time, the secrecy threshold cannot be lower than the finalization threshold; otherwise, the \tc output could be revealed before \mbbft finalizes a message. A naive way to address this is to increase the secrecy threshold to match the finalization threshold, i.e., $\ths=\thf$. However, the issue is that, since the threshold is ramped, $\thr>\ths=\thf=n-t$ (\mbbft protocols with optimal resilience typically have $\thf=n-t$ to ensure quorum intersection for safety), there may not be enough honest shares for parties to reconstruct \tc\ output upon finalizing the \mbbft\ output, violating the Termination property. 

We address this issue as follows:
First, we share the \tc\ secret $s$ among the parties twice, using independent randomness, with two different pairs of thresholds $(\ths,\thr)$ and $(\thr',\ths')$. 
The new thresholds are set to be  $\ths'=\thf, \thr'=\min(\ths'+(\thr-\ths), n)$.
Let $\{\sch{s}\}_{i\in [n]}$ and $\{\sch{s}'\}_{i\in [n]}$ be the secret shares of $s$ with thresholds $(\ths,\thr)$ and $(\ths',\thr')$, respectively. 
%
Second, we add a \emph{\fast}, where parties reveal their \tc\ output shares they compute with $\{\sch{s}'\}_{i\in [n]}$ immediately upon prefinalizing a message. 

Our final protocol is in~\Cref{alg:ramp}.
The setup phase is similar to~\Cref{alg:strawman}, except the same secret $s$ is shared twice, using independent randomness, for the \slow and \fast, respectively. 
Each party does the following:
\begin{itemize}[leftmargin=*]
    \item \emph{Fast path:} When a party prefinalizes a message $\m$ for \round \rd, it reveals its \tc output share $\peval(\sch{s}'_i,(\rd,\m))$. Once a party receives $\thr'$ verified shares of the \fast, it reconstructs the \tc output.
    \item \emph{Slow path:} Upon \mbbft\ finalization for message \m\ and \round \rd, each party $i$ reveals its \tc\ output share $\peval(\sch{s}_i,(\rd,\m))$. Next, any party who has not received $\thr'$ verified \tc\ output shares from the \fast waits to receive $\thr$ verified \tc\ output shares from the \slow. Once the party receives $\thr$ verified shares of the \slow, it reconstructs the \tc\ output.
\end{itemize}

Lastly, similarly to~\Cref{alg:strawman}, to guarantee Total Order, the parties push the finalized rounds into the FIFO $\queue$ and output the result round-by-round once either the \fast or \slow has reconstructed the \tc output. 
So the latency of the protocol is the minimum latency of the two paths.

Note that, a party reveals its share of the \slow even if it has revealed its share of the \fast or reconstructed the \tc output from the \fast. 
This is crucial for ensuring Termination, because with corrupted parties sending their \tc output shares to only a subset $\cS$ of honest parties, it is possible that only parties in $\cS$ can reconstruct the \tc output from the \fast. If honest parties in $\cS$ do not reveal their shares of the \slow, the remaining honest parties cannot reconstruct the \tc output, thereby losing the Termination guarantee. 
We defer the protocol analysis 
\ifsubmission
to the full version~\cite{xiang2024latency}
\else
to~\Cref{app:rand}
\fi
due to space constraints.

%% file: algos/ramp.tex
\begin{algorithm}[h!]
  \caption{Fast Path for Ramp Blockchain-Native Threshold Cryptosystem}
  \label{alg:ramp}
  \small
  \begin{minipage}[t]{\columnwidth}
    \medskip
    SETUP:
    \begin{algorithmic}[1]
        \State let $\bm{\m}\gets\{\}$, $\bm{\sigma}\gets \{\}$, $\queue\gets\{\}$
        \State let $(\{\pk_j\}_{j\in [n]},\sch{s}_i)\gets\sgen(s)$ for $t+1\leq \ths<\thr \leq n-t$
        \Comment{For slow path}
        \State let $(\{\pk'_j\}_{j\in [n]},\sch{s}'_i)\gets\sgen(s)$ for $\ths'=\thf, \thr'=\min(\ths'+(\thr-\ths), n)$
        \Comment{For fast path, $s$ is unknown to any party}
    \end{algorithmic}
    
    \medskip\hrule\medskip
    FAST PATH: 
    \begin{algorithmic}[1]
        \Upon{sending $(\PREFIN,\rd,\m)$ to all parties}
            \Comment{Augment the consensus protocol}
            \State let $\sigma'_i\gets \peval(\sch{s}'_i,(\rd,\m)) $
            \State {\bf send} $(\texttt{FAST-SHARE}, \rd,\m,\sigma'_i)$ to each other party
        \EndUpon
    \end{algorithmic}

    \medskip\hrule\medskip
    SLOW PATH: 
    \begin{algorithmic}[1]
        
        \Upon{$\fin(\rd,\m)$}
            \State let $\bm{\m}[\rd]\gets\m$ and $\sigma_i\gets \peval(\sch{s}_i,(\rd,\m)) $
            \State $\queue.push(\rd)$
            \State {\bf send} $(\texttt{SLOW-SHARE},\rd,\m,\sigma_i)$ to each other party 
        \EndUpon
    \end{algorithmic}

    \medskip\hrule\medskip
    RECONSTRUCTION:
    \begin{algorithmic}[1]
        \Upon{receiving $(\texttt{FAST-SHARE},\rd,\m,\sigma'_j)$ from party $j$}
        \Comment{Fast path reconstruction}
            \If{$\pver(\pk'_j,(\rd,\m),\sigma'_j) = 1$}
                \State $S'_{\rd,\m}\gets S'_{\rd,\m}\cup \{j\}$
                \If{$|S'_{\rd,\m}|\ge \thr'$ and $\bm{\sigma}[\rd]=\{\}$}
                    \State let $\bm{\sigma}[\rd]\gets \comb(S'_{\rd,\m},(\rd,\m),\{(\pk'_i,\sigma'_i)\}_{i\in S'_{\rd,\m}})$
                \EndIf
            \EndIf
        \EndUpon

        \Upon{receiving $(\texttt{SLOW-SHARE},\rd,\m,\sigma_j)$ from party $j$}
        \Comment{Slow path reconstruction}
            \If{$\pver(\pk_j,(\rd,\m),\sigma_j) = 1$}
                \State $S_{\rd,\m}\gets S_{\rd,\m}\cup \{j\}$
                \If{$|S_{\rd,\m}|\ge \thr$ and $\bm{\sigma}[\rd]=\{\}$}
                    \State let $\bm{\sigma}[\rd]\gets \comb(S_{\rd,\m},(\rd,\m),\{(\pk_i,\sigma_i)\}_{i\in S_{\rd,\m}})$
                \EndIf
            \EndIf
        \EndUpon
    \end{algorithmic}
    
    \medskip\hrule\medskip
    OUTPUT is same as~\Cref{alg:strawman}
  \end{minipage}
\end{algorithm}

%% file: section/evaluation.tex
\section{Distributed Randomness: A Case Study}
\label{sec:rand}
In this section, we implement and evaluate distributed randomness as a concrete example of \longname, to demonstrate the effectiveness of our solution in reducing the latency for real-world blockchains.
We implement the \fast (\Cref{alg:ramp}) for \posrand~\cite{das2024distributed}, which is a distributed randomness scheme designed for proof-of-stake blockchains and is deployed in the Aptos blockchain~\cite{aptos} to enable smart contracts to use randomness~\cite{randomness-api}. We then compare our latency (using both micro-benchmarks and end-to-end evaluation) with the~\cite{das2024distributed}, that implements the strawman protocol (\Cref{alg:strawman}).

In the rest of the section, we first provide a very brief overview of~\cite{das2024distributed} and our implementation of fast-path protocol atop their scheme. Due to space constraints, more details are deferred 
\ifsubmission
to the full version~\cite{xiang2024latency}.
\else
to~\Cref{app:rand}.
\fi

\para{Overview of \posrand~\cite{das2024distributed}.}
\posrand~\cite{das2024distributed} is a distributed randomness protocol for proof-of-stake blockchains where each party has a (possibly unequal) stake, and the blockchain is secure as long as the adversary corrupts parties with combined stake less than $1/3$-th of the total stake. 
Since the total stake in practice can be very large, \cite{das2024distributed} first assigns approximate stakes of parties to a much smaller value called {\em weights}, and this process is called {\em rounding}.
Parties in~\cite{das2024distributed} then participate in a publicly verifiable secret sharing~(PVSS) based distributed key generation~(DKG) protocol to receive secret shares of a \tc\ secret $s$. 
After DKG, \posrand~\cite{das2024distributed} implements the weighted extension of \slow (\Cref{alg:strawman}) for \longname (\Cref{def:btc}), where they use a distributed verifiable unpredictable function~(VUF) as the \tc\ protocol. More precisely, for each finalized block, each party computes and reveals its VUF shares. Next, once a party receives verified VUF shares from parties with combined weights greater than or equal to $w$, it reconstructs the VUF output.

\para{Implementation of \fast (\Cref{alg:ramp}).}
Recall that the \fast requires sharing the same secret with two sets of thresholds, i.e., $\ths<\thr$ for the \slow and $\ths'<\thr'$ for the \fast. 
Consequently, we augment the rounding algorithm of~\cite{das2024distributed} to additionally take $(\ths',\thr')$ as input, and output the weight threshold for the \fast. 
To setup the secret-shares of the \tc\ secret, we use the DKG protocol of~\cite{das2024distributed} with the following minor modifications. Each party starts by sharing the same secret independently using two weight thresholds $w$ and $w'$. The rest of the DKG protocol is identical to~\cite{das2024distributed}, except parties agree on two different aggregated PVSS transcript instead of one. Note that, these doubles the computation and communication cost of DKG.
As described in~\Cref{alg:ramp}, parties reveal their VUF shares (\tc output shares) for the \fast upon prefinalizing a block, and for the \slow upon finalizing a block.
For both paths, the parties collect the VUF shares and are ready to execute the block as soon as the randomness (\tc output) is reconstructed from either path. 

\subsection{Evaluation Results.}
\label{sec:rand:eval}

\input{section/eval_table}

We implement our \fastrand protocol (\Cref{alg:ramp}) in Rust, atop the open-source \posrand~\cite{das2024distributed} implementation~\cite{aptos} on the \aptos blockchain.
We worked with \aptos Labs to deploy and evaluate our protocol on the \aptos mainnet.

\para{Setup and metrics.}
As of July 2024, the \aptos blockchain is run by $140$ validators, distributed $50$ cities across $22$ countries with the stake distributed described in~\cite{aptos-geo}. 
The $50$-th, $70$-th and $90$-th percentile (average) of round-trip latency between the blockchain validators is approximately $150$ms, $230$ms, and $400$ms, respectively.
Due to space limitation, we defer other details of the evaluation setup 
\ifsubmission
to the full version~\cite{xiang2024latency}.
\else
to~\Cref{app:rand:eval}.
\fi
We measure the {\em randomness latency} as the duration required to generate randomness for each block, as in~\Cref{def:latency}.
It measures the duration from the moment the block is finalized by consensus to the when the randomness for that block becomes available. 
We report the average randomness latency (measured over a period of 12 hours).
We also measure and compare the setup overhead for \posrand~\cite{das2024distributed} and \fastrand, using micro-benchmarks on machines of the same hardware specs as the \aptos mainnet.

\one
\para{Randomness latency.} \Cref{tab:rand} summarizes the latency comparison of our \fastrand and \posrand~\cite{das2024distributed}. As observed, \fastrand significantly reduces the randomness latency of \posrand by 71\%, from 85.5 ms to 24.7 ms.
As mentioned in~\Cref{sec:ramp:opt}, the small latency overhead of \fastrand comes from the fact that honest parties may need to wait slightly longer after local finalization to receive additional shares from the \fast, since the reconstruction threshold of the \fast is higher than the finalization threshold. 
To show the significance of the latency improvement for consensus end-to-end latency, we also measure the consensus latency 362 ms as the duration of each block from proposed to finalized. As shown in the table, \fastrand improves the latency overhead from $23.6\%$
to $6.8\%$ in terms of the consensus latency.
The latency of the \slow in our \fastrand protocol (\Cref{alg:ramp}) is comparable to that of \posrand~\cite{das2024distributed} and is therefore omitted for brevity. 

%% file: section/eval_table.tex
\begin{table}[t!]
    \centering
    \small
    \setlength\tabcolsep{3pt}
    \begin{tabular}{cccccc}
        \toprule
        Scheme & \makecell{Randomness \\ Latency} & \makecell{Consensus \\ Latency} & \makecell{Overhead on top \\ of Consensus}  \\
        \midrule
        \posrand~\cite{das2024distributed} & 85.5 ms & 362 ms & 23.6\% \\
        our \fastrand & 24.7 ms & 362 ms & 6.8\% \\
        \bottomrule
    \end{tabular}
    \vspace{1mm}
    \caption{Latencies of \posrand~\cite{das2024distributed} and our \fastrand on \aptos mainnet (July 2024).}
    \vspace{-3mm}
    \label{tab:rand}
\end{table}

%% file: appendix/related.tex
\section{Related Work}
\label{app:related}

\para{Latency in blockchains.}
Latency reduction in blockchain has been an important problem for decades. Many works focused on reducing the consensus latency under partial synchrony, for leader-based BFT protocols~\cite{martin2006fast,abraham2021good,gueta2019sbft,gelashvili2022jolteon,jalalzai2020fast,doidge2024moonshot} and DAG-based BFT protocols~\cite{spiegelman2022bullshark,spiegelman2023shoal,arun2024shoal++,keidar2022cordial,babel2023mysticeti}. 
Several other works add a fast path to the consensus protocol, allowing certain transactions to be finalized faster~\cite{baudet2020fastpay,baudet2023zef,sui}.
Research has also focused on improving the latency in consensus protocols in various optimistic scenarios. 
Optimistic BFT protocols aim to achieve small latencies under certain optimistic conditions~\cite{abraham2020sync,dutta2005best,song2008bosco,kotla2007zyzzyva}.
For instance, the optimistic fast paths in~\cite{abraham2020sync,shrestha2020optimality} require more than $3n/4$ parties to be honest in synchrony, while those in~\cite{kotla2007zyzzyva,gueta2019sbft} require all $n$ parties to vote under partial synchrony.

\one
\para{Threshold cryptography in blockchains.}
In many applications of modern blockchains, threshold cryptography plays a vital role. 
An important example of \longname is to generate distributed randomness for blockchains~\cite{das2024distributed,hanke2018dfinity}.
where existing real-world deployments~\cite{das2024distributed,hanke2018dfinity,aptos,sui,flow} follow the threshold VRF-based approach.
Numerous blockchain research \\
\cite{reiter1994securely,yang2022sok,asayag2018fair,zhang2023f3b,momeni2022fairblock,kavousi2023blindperm,malkhi2022maximal,bebel2022ferveo} focuses on MEV countermeasures and privacy enhancement using threshold decryption, where transactions are encrypted with threshold decryption and only revealed and executed by blockchain parties after finalization. 
To the best of our knowledge, all existing constructions supporting the \longname mentioned above incur an additional round of latency, except~\cite{bebel2022ferveo}. 
In~\cite{bebel2022ferveo}, threshold decryption is employed to mitigate MEV while using Tendermint as the consensus protocol, with parties revealing their decryption shares when voting for a block. However, the described approach is insecure, as an adversary could decrypt the block before it is finalized. This vulnerability arises from an incorrect description of the Tendermint consensus mechanism, which mistakenly assumes only a single round of consensus voting.
Furthermore, their approach is less general compared to the results presented in this paper.

%% file: appendix/prelim.tex
\section{Preliminaries}
\label{app:prelim}

\subsection{Cryptography Definitions}
\label{app:prelim:crypto}
We formalize the robustness property using the \rbcma\ game in Game~\ref{game:rbcma}. Intuitively, the robustness property ensures that the protocol behaves as expected for honest parties, even in the presence of an adversary that corrupts up to $t$ parties.
More precisely, it says that:
\begin{enumerate*}[label=(\roman*)]
    \item $\pver$ should always accept honestly generated \tc output shares and
    \item if we combine $\thr$ valid \tc output shares (accepted by $\pver$) using the $\comb$ algorithm, the output of $\comb$ should be equal to $\eval(s,\val)$, except with a negligible probability. 
\end{enumerate*}
The latter requirement ensures that maliciously generated \tc output share cannot prevent honest parties from efficiently computing $\eval(s,\val)$ (except with a negligible probability). Note that we allow \adv\ to generate \tc output share arbitrarily. Also, we can achieve robustness even if \adv\ learns shares of all parties.

\begin{game}[Robustness Under Chosen Message Attack]
\label{game:rbcma}
For a $(n,\ths,\thr)$-threshold cryptosystem \tc\ we define the game \rbcma\ in the presence of adversary \adv\ as follows:
\begin{itemize}[leftmargin=*]
    \item {\bf Setup.} $\adv$ specifies a set $\cC\subset [n]$, with $|\cC|< \ths$ of corrupt parties. Let $\cH :=[n]\setminus\cC$ be the set of honest parties. 
    \one
    \item {\bf Share generation}. Run $\sgen(s)$ to generate the shares of $s$. $\adv$ learns $\sch{s}_i$ for each $i\in\cC$ and all the public keys $\{\pk_1,\ldots,\pk_n\}$.
    \one
    \item {\bf Function evaluation shares.} \adv\ submits a tuple $(i,\val)$ for some $i\in\cH$ and $\val \in \cX$ as input and receives $\sigma_i\gets \peval(\sch{s}, \val)$.
    \one
    \item {\bf Output determination.}  Output 1 if either of the following happens; otherwise, output 0.
    \begin{enumerate}
        \item \adv\ outputs $(i,\val)$ such that $\pver(\pk_i,\peval(\sch{s}_i,\val))=0$;
        \item \adv\ outputs $(S,\{\sigma_i\}_{i\in S},\val)$ where $S\subseteq [n]$ with $|S|\ge \thr$ and $\pver(\pk_i,\sigma_i,\val)=1$ for all $i\in S$, such that $\comb(S,\{\pk_i,\sigma_i\},\val)\ne \eval(s,\val)$.
    \end{enumerate} 
\end{itemize}
\end{game}

\begin{definition}[Robustness Under Chosen Message Attack]
\label{def:tcai-robust}
Let $\tc=(\setup,\sgen,\eval,\peval,\pver,\comb,\ver)$ be a $(n,\ths,\thr)$-threshold cryptosystem. Consider the game $\rbcma$ defined in Game~\ref{game:rbcma}. We say that \tc\ is $\rbcma$ secure, if for all PPT adversaries \adv, the following advantage is negligible, i.e.,
\begin{equation}
    \Pr[\rbcmap\Rightarrow 1] =\negl(\lambda)
\end{equation}
\end{definition}

Next, we describe the secrecy property. Intuitively, the secrecy property ensures that $\eval(s,\val)$ remains hidden from an adversary $\adv$\ that corrupts up to $\ths$ parties, where the precise notion of ``hidden'' depends on the application. 
For example, when \tc\ is a threshold decryption scheme, i.e., $\eval(s,\val)$ outputs the decryption of the ciphertext $\val$ using secret key $s$, the Secrecy property requires that \tc\ is semantically secure in the presence of an attacker \adv\ that corrupts up to $t$ parties.

Looking ahead, we will use a distributed randomness beacon as our concrete application (see~\Cref{app:rand}), with unpredictability under chosen message attack as our secrecy property. Intuitively, the unpredictability property ensures that an adversary corrupting less than $\ths$ parties can not compute $\eval(s,\val)$ for a value $\val$ for which it has seen less than $\ths-|\cC|$ \tc output shares. We formalize this with the \upcma\ game in Game~\ref{game:upcma}. 

%
\begin{game}[Unpredictability Under Chosen Message Attack]
\label{game:upcma}
For a $(n,\ths,\thr)$-threshold cryptosystem \tc\ the game \upcma\ in the presence of adversary \adv\ as follows:
\begin{itemize}[leftmargin=*]
    \item {\bf Setup.} $\adv$ specifies two sets $\cC,\cS\subset [n]$, with $|\cC \cup \cS|< \ths$. Here, $\cC$ is the set of corrupt parties and $\cS$ is the set of honest parties that $\adv$ queries for \tc output shares on the forged input,. Let $\cH :=[n]\setminus\cC$ be the set of honest parties.  
    \item The share generation, and function evaluation shares steps are identical to the \rbcma\ game. 
    \item {\bf Output determination.}  \adv\ outputs $(\val^*,\eval(s,\val^*))$. Output 1 if \adv\ has queried for \tc output share on $\val^*$ from only parties in $\cS$. Otherwise, output $0$.
\end{itemize}
\end{game}

\begin{definition}[Unpredictability Under Chosen Message Attack]
\label{def:tcai-unpred}
Let $\tc= (\setup,\sgen,\eval,\peval,\pver,\comb,\ver)$ be a $(n,\ths,\thr)$-threshold cryptosystem. Consider the game $\upcma$ in Game~\ref{game:upcma}. We say that \tc\ is $\upcma$ secure, if for all PPT adversaries \adv, the following advantage is negligible, i.e.,
\begin{equation}
    \Pr[\upcmap\Rightarrow 1] =\negl(\lambda)
\end{equation}
\end{definition}

\section{Blockchain-Native Threshold Cryptosystem}
\label{app:btc}

The following lemma is a direct implication of the \longname definition (\Cref{def:bntc}). 
Intuitively, a \longname should not reconstruct the \tc output before a message is finalized. 

\begin{lemma}\label{lem:impos}
    For any \longname protocol, for any execution, and for all \round $r\in \N$, we have $\grtr\geq\gftr$. 
\end{lemma}
\begin{proof}
    Let $\name=(\mbb,\tc)$ denote the protocol. 
    Suppose that $\grtr<\gftr$ in some execution. 
    By the definition of $\gftr$ and $\grtr$, the adversary learns the \tc output of \round \rd at time $\grtr$, before the message is globally finalized for \rd by $\mbb$ at time $\gftr$. According to~\Cref{def:gft}, the system is in multivalent state for \rd at time $\grtr$, which means there exist two execution extensions of $\mbb$ where the honest parties output different messages for \rd. Then, the Secrecy (\Cref{def:bntc}) of \name is violated for at least one of the execution, contradiction.
\end{proof}

%% file: appendix/tight.tex
\section{Tight Blockchain-Native Threshold Cryptosystem}
\label{app:tight}

\subsection{Analysis of~\Cref{alg:tight}}
%
%
\begin{theorem}\label{thm:tight:correctness}
    \Cref{alg:tight} implements a \longname and guarantees the Agreement, Termination, Validity, Total Order, and Secrecy properties. 
\end{theorem}
\begin{proof}
    Let $\nameft=(\mbbft,\tc)$ denote the protocol.

    \paragraph{Secrecy.} 
    We first prove that, for any \round $\rd\in\N$, if an honest party outputs $(\rd,\m,\sigma)$ then $\sigma=\eval(s,(\rd,\m))$. 
    By the Agreement and Termination properties of \mbbft, all honest parties output $\fin(\rd,\m)$ in \mbbft.
    By~\cref{def:ft}, 
    \begin{enumerate}
        \item There exist $\thf$ parties (or equivalently $\thf-|\cC|$ honest parties) that have sent $(\PREFIN,\rd,\m)$.
        \item If $\m\neq\bot$, then no honest party sent $(\PREFIN,\rd,\m')$ such that $\m'\in\cM$ and $\m'\ne\m$. 
        \item If $\m=\bot$, then there exists at most one $\m'\in\cM$ such that any honest party may send $(\PREFIN,\rd,\m')$. 
    \end{enumerate}
    
    For the sake of contradiction, suppose that for some $\rd\in\N$, an honest party $h$ outputs $(\rd,\m,\sigma')$ where $\sigma'\ne\eval(s,(\rd,\m))$. 
    As per the protocol and the Robustness property of \tc, $\sigma'=\eval(s,(\rd,\m'))$ for some $\m'\ne\m$. 
    Note that in \Cref{alg:tight} an honest party reveals its \tc\ output share for any $(\rd,\m')$ only upon prefinalizing $(\rd,\m')$ or finalizing $(\rd,\m')$. 
    This implies, by the Unpredictability property of \tc, that at least $\ths-|\cC|$ honest parties have revealed their \tc\ output shares  $\peval(\sch{s},(\rd,\m'))$. Let $T_{m'}$ be the set indicating these honest parties. 
    \begin{itemize}
        \item If any honest party $i\in T_{m'}$ finalizes a message $(\rd,\m')$ for $m'\ne m$, then this violates the Agreement property of \mbbft, and hence a contradiction. 
        \item This implies that all parties in $T_{m'}$ has sent $(\PREFIN,\rd,\m')$ in \mbbft. If $\m\neq\bot$, then no honest party should have sent $(\PREFIN,\rd,\m')$ such that $\m'\in\cM$ and $\m'\ne\m$, contradiction. If $\m=\bot$, then all $|T_{m'}|\geq \ths-|\cC|=\thf-|\cC|$ honest parties in $T_{m'}$ has sent $(\PREFIN,\rd,\m')$ in \mbbft. This implies that the message $\m'$ is globally finalized, again violating the Agreement property of \mbbft, hence a contradiction. 
    \end{itemize}
    Therefore, for any \round $\rd\in\N$, if an honest party outputs $(\rd,\m,\sigma)$ then $\sigma=\eval(s,(\rd,\m))$.

    For any corrupted party, the same argument above applies; thus the adversary cannot learn $\eval(s,(\rd,\m'))$ where $\m'\ne\m$.
    
    \paragraph{Agreement.} 
    Suppose an honest party outputs $(\rd,\m,\sigma)$ and another honest party outputs $(\rd,\m',\sigma')$. 
    The Agreement property of \mbbft\ ensures that $\m=\m'$. Therefore, by the Secrecy property of \nameft\ above, we get $\sigma=\sigma'$.

    \paragraph{Termination.} 
    The Termination property of \mbbft\ requires that $\thf\leq n-t$, since the honest party needs to finalize the message even when the corrupted parties all remain slient. 
    After GST, by the Agreement and Termination property of \mbbft, for any \round $\rd\in\N$ all honest parties eventually output the same $\fin(\rd,\m)$ where $\m\in\cM\cup\{\bot\}$. 
    Then, eventually, all $n-t$ honest parties send their $\peval(\sch{s},(\rd,\m))$ to all parties according to the prefinalization step of the protocol. 
    By the Robustness property of \tc, all honest parties can eventually reconstruct $\sigma=\eval(s,(\rd,\m))$ since $\thr=\thf\leq n-t$.
    Therefore, after GST, each honest party eventually outputs $(\rd,\m,\sigma)$ for some $\sigma$. 
    
    \paragraph{Validity.} 
    By the Validity property of \mbbft, all honest parties eventually output $\fin(\rd,\m)$ in \mbbft, and will eventually output $(\rd,\m,\sigma')$ by Termination of \name.
    Since if any honest party outputs $(\rd,\m,\sigma')$ then $\sigma'=\eval(s,(\rd,\m))$ by the Secrecy property of \name, we conclude that all honest parties eventually output $(\rd,\m,\eval(s,(\rd,\m)))$. 

    \paragraph{Total Order.} 
    In the protocol, an honest party always outputs $(\rd,\m,\sigma)$ according to the order of $\rd$ in the FIFO $\queue$. 
    Suppose that an honest party outputs $(\rd,\m,\sigma)$ before $(\rd',\m',\sigma')$ in \name, then the party enqueues $\rd$ before $\rd'$. This implies that the honest party outputs $\fin(\rd,\m)$ before $\fin(\m',\rd')$ in \mbbft. By the Total Order property of \mbbft, we conclude $\rd<\rd'$ since an honest party outputs $\fin(\rd,\m)$ before $\fin(\m',\rd')$. Therefore, the protocol satisfies Total Order: if an honest party outputs $(\rd,\m,\sigma)$ before $(\rd',\m',\sigma')$, then $\rd<\rd'$.

    \qed
\end{proof}

\begin{theorem}
    \Cref{alg:tight} achieves $\gftr=\grtr$ for any \round $\rd\in\N$ in any execution. 
\end{theorem}
\begin{proof}
    Recall from~\Cref{sec:prelim}, for any \round\ \rd, $\gftr$ and $\grtr$ are the global finalization time and global reconstruction time of \round \rd, respectively.
    Consider any execution. 
    By the definition of $\gftr$, $\thf-|\cC|$ honest parties have sent $(\PREFIN,\rd,\m)$ at time $\gftr$.
    According to~\Cref{alg:tight}, these honest parties have also revealed their \tc output shares for $(\rd,\m)$ at time $\gftr$. Since, $\thr=\thf$, this implies that the adversary can  adversary to reconstruct the \tc output at time $\gftr$, and hence $\grtr\leq\gftr$. 
    From~\Cref{lem:impos}, we have that $\grtr\geq \gftr$.
    Therefore, we conclude that $\gftr=\grtr$ for any $\rd$ in any execution. 

    \qed
\end{proof}

\begin{theorem}\label{thm:tight:latency}
    \Cref{alg:tight} achieves $\latency_{\rd}=0$ for any \round $\rd\in\N$ in any error-free execution. 
\end{theorem}
\begin{proof}
    Recall from~\Cref{sec:prelim}, in any \round\ \rd, $\latency_\rd$ is the maximum time difference, across all honest parties, between the time a honest party $i$ finalizes $(\m,\rd)$ from \mbb\  and the time the same honest party $i$ outputs $(\m,\rd,\eval(s,(\m,\rd))$.
    For any fixed \round\ \rd, consider any error-free execution. 
    By the definition of \mbbft, an honest party finalizes a message once it receives \PREFIN\ messages from $\thf$ parties.
    Since all parties in an error-free execution are honest, these parties also send their \tc\ output share $\peval(\sch{s}_i,(\rd,\m))$ upon prefinalizing $(\rd,\m)$. 
    Since $\thf=\thr$, any honest party simultaneously receives $\thr$ valid shares of the form $\peval(\sch{s},(\rd,\m))$ and $\thf$ $\PREFIN$ messages for \mbbft. Therefore, every honest party finalizes $(m,r)$ and computes $\bm\sigma[\rd]=\eval(\sch{s},(\rd,m))$ for any \round \rd, simultaneously.

    Now we prove the theorem by induction on the \round number $\rd$.
    For the base case, consider $\rd=0$. We have shown that any honest party computes $\bm\sigma[\rd]$ and finalizes $(\rd,\m)$ at the same time, it also outputs $(\rd,\bm{\m}[\rd],\bm\sigma[\rd])$ at the same time since $\rd=0$ is at the top of $\queue$ by the Total Ordering property of \mbbft. Thus, $\latency_{0}=0$.
    For the induction steps, assume that the theorem is true up to \round $\rd=k-1$, that is, $\latency_{\rd}=0$ for $\rd=0,...,k-1$.
    Now consider \round $\rd=k$. Similarly, any honest party computes $\bm\sigma[\rd]$ and finalizes $(\rd,\m)$ at the same time. Due to the Total Ordering property of \mbbft, any rounds pushed in $\queue$ before $\rd=k$ must be smaller than $k$, and they have been popped when $(\rd,\m)$ is finalized since $\latency_{\rd}=0$ for $\rd=0,...,k-1$. Therefore, $\rd=k$ is at the top of $\queue$ and the party can output immediately, thus $\latency_{k}=0$.
    Therefore, by induction, the theorem holds. 
    \qed
\end{proof}

%% file: appendix/ramp.tex
\section{Ramp Blockchain-Native Threshold Cryptography}
\label{app:ramp}

\subsection{Impossibility}
\label{app:ramp:impos}

\begin{proof}[Proof of~\Cref{thm:impos}]
    From~\Cref{lem:impos}, we have that in all executions, in each round $\rd \in \N$,  $\grtr\geq\gftr$. 
    
    For the sake of contradiction, suppose that there exists a protocol $\nameft=(\mbbft,\tc)$ with ramp thresholds, such that for any given \round\ \rd, in all executions we have $\grtr=\gftr$. 
    Let $\cE$ be one such execution and let $\tau_{\cE}=\grtr=\gftr$. 
    Let $s$ denote the secret shared between the parties. 
    Since \mbbft has the finalization threshold $\thf$, at least $\thf-|\cC|$ honest parties have sent $(\PREFIN,\rd,\m)$ at $\tau_{\cE}$.
    Similarly, by the definition of $\grtr$, at least $\thr-|\cC|$ honest parties have revealed their \tc\ output shares $\peval(\sch{s},(\rd,\m))$ at $\tau_{\cE}$. 
    Without loss of generality, we can assume that exactly $\thf-|\cC|$ honest parties have sent $(\PREFIN,\rd,\m)$ and $\thr-|\cC|$ honest parties have revealed their \tc output shares at $\tau_{\cE}$. 
    Let $h$ be any honest party that sent $(\PREFIN,\rd,\m)$ at time $\tau_{\cE}$ in the execution $\cE$. 
    There are two cases, which we denote with (1) and (2) below.
    \begin{enumerate}
        \item[(1)] $h$ also reveals its \tc output share $\peval(\sch{s}_h, (\rd,\m))$ at time $\tau_{\cE}$. 
        Consider another execution $\cE'$ identical to $\cE$ up to time $\tau_{\cE}$ with the only difference that, $h$ sends $(\PREFIN,\rd,\m)$ and reveals $\peval(\sch{s}_h,(\rd,\m))$ at time time $\tau_{\cE}+\epsilon$ for some $\epsilon>0$ (say due to asynchrony in computation or communication). Therefore, $\tau_{\cE}+\epsilon$ is the global finalization time of execution $\cE'$. 
        
        Now, consider the time $\tau_{\cE}$ in the new execution $\cE'$.
        The message $\m$ is not globally finalized at $\tau_{\cE}$ for $\rd$ since only $\thf-|\cC|-1$ honest parties have sent $(\PREFIN,\rd,\m)$.
        However, since $\thr>\ths$, there are $\thr-|\cC|-1\geq \ths-|\cC|$ honest parties that have revealed their \tc output shares $\peval(\sch{s},(\rd,\m))$. This implies that in execution $\cE'$ the adversary \adv\ can compute the \tc output $\eval(s,(m,r))$ at $\tau_{\cE}$, which is earlier than its global finalization time $\tau_{\cE}+\epsilon$. This contradicts \Cref{lem:impos}.
        
        \item[(2)] $h$ does not reveal its \tc output share at time $\tau_{\cE}$.
        Similarly, consider another execution $\cE'$ identical to $\cE$ up to time $\tau_{\cE}$ with the only difference that, $h$ sends $(\PREFIN,\rd,\m)$ at time $\tau_{\cE}+\epsilon$ for some $\epsilon >0$ (say due to asynchrony in computation or communication). Therefore, for the new execution $\cE'$, $\tau_{\cE}+\epsilon$ is the global finalization time, and the global reconstruction time remains $\tau_{\cE}$. 
        Then, for the new execution $\cE'$, the global reconstruction comes before the global finalization, contradicting~\Cref{lem:impos}. 
        \qed
    \end{enumerate}
\end{proof}

\subsection{Analysis of~\Cref{alg:ramp}}
\begin{theorem}\label{thm:ramp:correctness}
    \Cref{alg:ramp} implements a \longname and guarantees the Agreement, Termination, Validity, Total Order, and Secrecy properties. 
\end{theorem}
\begin{proof}
    Let $\nameft=(\mbbft,\tc)$ denote the protocol. The proofs of Agreement, Termination, Validity and Total Order properties are identical to that of \Cref{thm:tight:correctness}. Thus, we focus on the Secrecy property.
    
    \paragraph{Secrecy.}
    We first prove that, for any \round $\rd\in\N$, if an honest party outputs $(\rd,\m,\sigma)$ then $\sigma=\eval(s,(\rd,\m))$.
    By the Agreement and Termination properties of \mbbft, all honest parties output $\fin(\rd,\m)$ in \mbbft.
    By~\cref{def:ft}, 
    \begin{enumerate}
        \item There exist $\thf$ parties (or equivalently $\thf-|\cC|$ honest parties) that have sent $(\PREFIN,\rd,\m)$.
        \item If $\m\neq\bot$, then no honest party sent $(\PREFIN,\rd,\m')$ such that $\m'\in\cM$ and $\m'\ne\m$. 
        \item If $\m=\bot$, then there exists at most one $\m'\in\cM$ such that any honest party may send $(\PREFIN,\rd,\m')$. 
    \end{enumerate}

    For the sake of contradiction, suppose that for some $\rd\in\N$, an honest party $h$ outputs $(\rd,\m,\sigma')$ where $\sigma'\ne\eval(s,(\rd,\m))$. 
    As per the protocol and the Robustness property of \tc, $\sigma'=\eval(s,(\rd,\m'))$ for some $\m'\ne\m$. 
    Also, as per the protocol and the Robustness property of \tc, $\sigma'=\eval(s,(\rd,\m'))$ for some $\m'\ne\m$. 
    There are two cases:
    
    First, if $\sigma'$ is obtained from the \slow, then by the Unpredictability property of \tc, at least $\ths-|\cC|\ge 1$ honest parties have revealed their $\peval(\sch{s},(\rd,\m'))$ for the \slow. According to the protocol, these honest parties output $\fin(\rd,\m')$, which violates the Agreement property of \mbbft since an honest party outputs $\fin(\rd,\m)$, contradiction.
    
    Second, if $\sigma'$ is obtained from the \fast, then by the Unpredictability property of \tc, at least $\ths'-|\cC|$ honest parties reveal their $\peval(\sch{s}',(\rd,\m'))$ for the \fast. 
    If $\m\neq\bot$, then no honest party should have sent $(\PREFIN,\rd,\m')$ such that $\m'\in\cM$ and $\m'\ne\m$, contradiction. If $\m=\bot$, then at least $\ths'-|\cC|=\thf-|\cC|$ honest parties has sent $(\PREFIN,\rd,\m')$ in \mbbft. This implies that the message $\m'$ is globally finalized, again violating the Agreement property of \mbbft, hence a contradiction. 
    Therefore, for any \round $\rd\in\N$, if an honest party outputs $(\rd,\m,\sigma)$ then $\sigma=\eval(s,(\rd,\m))$.

    For any corrupted party, the same argument above applies; thus the adversary cannot learn $\eval(s,(\rd,\m'))$ where $\m'\ne\m$ 
    \qed
\end{proof}

\begin{theorem}
    \Cref{alg:ramp} achieves $\gftr=\grtr$ for \round $\rd$ if the execution is synchronized for $\rd$. 
\end{theorem}
\begin{proof}
    Consider any synchronized execution. 
    By the definition of synchronized execution, all honest parties prefinalize the same message $\m$ for \round\ $\rd$ at the same physical time $\gftr$. 
    Thus, as per ~\Cref{alg:ramp}, all honest parties reveal their \tc output shares for $(\rd,\m)$ at time $\gftr$ in the \fast, which allows the adversary to reconstruct the \tc output at time $\gftr$, which implies $\grtr\leq\gftr$. 
    By~\Cref{lem:impos}, $\grtr\geq \gftr$.
    Therefore, we conclude that $\gftr=\grtr$ for any \rd in any synchronized execution. 
    \qed
\end{proof}

\begin{theorem}\label{thm:ramp:latency}
    \Cref{alg:ramp} achieves $\latency_{\rd}=0$ for any \round $\rd\in\N$ in any optimistic execution.
\end{theorem}
\begin{proof}
    Consider any optimistic execution and any \round $\rd\in\N$.
    As per \Cref{alg:ramp}, honest parties prefinalize a message and send its \tc output share for the \fast simultaneously. Recall from the definition of optimistic execution, in optimistic executions, all messages have the same delay, and all parties are honest and finalize the same message at the same physical time. Therefore, all honest parties locally receive $\thf$ $\PREFIN$ messages and $\thr$ \tc\ output shares for $(\rd,\m)$ simultaneously. This implies that the local finalization and reconstruction occur simultaneously. 
    Then, by a similar induction argument as in~\Cref{thm:tight:latency}, $\latency_{\rd}=0$ for any \round $\rd\in\N$ in any optimistic execution.
    \qed
\end{proof}

In practice, different honest parties may prefinalize at different physical times due to a lack of synchrony. Moreover, honest parties may need to wait slightly longer after local finalization to receive additional shares from the \fast since the reconstruction threshold of the \fast is higher than the finalization threshold. 
Despite this,  the latency of~\Cref{alg:ramp} remains significantly lower than one message delay. The evaluation in~\Cref{app:rand:eval} demonstrates that the \fast reduces the latency overhead by $71\%$ compared to the \slow (\Cref{alg:strawman}), which has a latency of one message delay.

%% file: appendix/evaluation.tex
\section{Distributed Randomness: A Case Study}
\label{app:rand}

\subsection{Overview of \posrand~\cite{das2024distributed}}
\label{app:rand:overview}
\posrand~\cite{das2024distributed} is a distributed randomness protocol for proof-of-stake blockchains where each party has a stake, and the blockchain is secure as long as the adversary corrupts parties with combined stake less than $1/3$-th of the total stake. 

\para{Rounding.}
Note that the total stake in practice can be very large. For example, as of July 2024, the total stake in the Aptos blockchain~\cite{aptos-stake} exceeds $8.7\times 10^8$. Since a distributed randomness protocol with such a large total stake will be prohibitively expensive,\cite{das2024distributed} assigns approximate stakes of parties to a much smaller value called {\em weights}, and this process is called {\em rounding}~\footnote{The rounding is only used for the \tc protocol. The \mbb (consensus) protocol is still based on accurate stakes.}.
Briefly, the rounding algorithm in~\cite{das2024distributed} defines a function ${\sf Round}(\bm{S}, \ths, \thr)\rightarrow (\bm{W}, w)$ that inputs the stake distribution $\bm{S}$ of the parties before rounding, the secrecy and reconstruction threshold $\ths, \thr$ (in stakes), and outputs the weight distribution $\bm{W}$ of the parties after rounding, and the weight threshold $w$. 
The rounding algorithm guarantees that
any subset of parties with a combined stake $<\ths$ will always have a combined weight $< w$, thus preventing them from reconstructing the \tc output; and
any subset of parties with a combined stake $\geq\thr$ will always have a combined weight $\geq w$, allowing them to reconstruct the \tc output. We refer the reader to~\cite{das2024distributed} for the concrete implementation of the ${\sf Round}$ algorithm.

Parties in~\cite{das2024distributed} then participates in a publicly verifiable secret sharing~(PVSS) based distributed key generation~(DKG) protocol to receive secret shares of a \tc\ secret $s$. 

\one
\para{Randomness generation.}
\posrand~\cite{das2024distributed} implements the weighted extension~\footnote{The weighted extension of \Cref{alg:strawman} is where each party has a stake and the threshold check (line 4 of the reconstruction phase) is based on the stake sum instead of the number of parties. However, the implementation can check the combined weight against the weight threshold $w$, instead of checking the stakes. } of \slow (\Cref{alg:strawman}) for \longname (\Cref{def:btc}), where they use a distributed verifiable unpredictable function~(VUF) as the \tc\ protocol. More precisely, for each finalized block, each party computes and reveals its VUF shares. Next, once a party receives verified VUF shares from parties with combined weights greater than or equal to $w$, it reconstructs the VUF output as its \tc\ output.

\subsection{Implementation}
\label{app:rand:impl}
Now we describe our implementation of \fast (\Cref{alg:ramp}) atop~\cite{das2024distributed}.

\one
\para{Consensus.}
Here we explain how the consensus in~\cite{das2024distributed} satisfies the definition of \mbbft (\Cref{def:ft}).
\posrand~\cite{das2024distributed} is built on the \aptos blockchain, which uses a latency-improved version~\cite{optjolteon} of \jolteon~\cite{gelashvili2022jolteon} as its consensus protocol. 
\jolteon tolerates an adversary capable of corrupting parties holding up to $t$ stakes out of a total of $n=3t+1$ stakes under partial synchrony. 
In \jolteon, each message $\m$ in \mbbft is a block containing transactions, and only one party (called leader) can propose a block for each round.
\begin{itemize}
    \item {\em Prefinalization.}
    A party at local round $\leq \rd$ sends $(\PREFIN,\rd,\m)$ upon receiving a quorum certificate (QC) for $\m$ of round $\rd$, where a QC consists of votes from parties with combined stakes of $2/3$ of the total stakes.
    When a party prefinalizes a block $\m$ of round $\rd$, it also implicitly prefinalizes all previous rounds that are not yet finalized (either parent blocks or $\bot$).
    In \jolteon, each party votes for at most one block per round, so that at most one QC can be formed for each round by quorum intersection.
    This means that for each round, all honest parties can only prefinalize at most one block. 
    \item {\em Finalization.} 
    A block $\m$ is globally finalized for round $\rd$ if and only if parties with combined $\thf=2t+1$ stakes (or $\thf-t=t+1$ honest stakes) send $(\PREFIN,\rd,\m)$. 
    A party locally finalizes $\m$ for round $\rd$ once it receives $(\PREFIN,\rd,\m)$ messages from parties with $\thf$ stakes, which also finalizes all previous rounds with parent blocks or $\bot$.
\end{itemize}

\para{Rounding.}
Recall that the \fast (\Cref{alg:ramp}) requires sharing the same secret two sets of thresholds, i.e., $\ths<\thr$ for the \slow and $\ths'<\thr'$ for the \fast. 
Consequently, we augment the rounding algorithm of~\cite{das2024distributed} to additionally take $\ths',\thr'$ as input, and output the weight threshold $w'$ for the \fast. 
That is, ${\sf Round'}(\bm{S}, \ths, \thr, \ths', \thr')\rightarrow (\bm{W}, w, w')$ where $w,w'$ are the weight thresholds of the \slow and \fast, respectively. ${\sf Round'}$ also provides the same guarantees for the thresholds of \fast, namely 
any subset of parties with a combined stake $<\ths'$ will always have a combined weight $< w'$, and
any subset of parties with a combined stake $\geq\thr'$ will always have a combined weight $\geq w'$.
As in~\Cref{alg:ramp},  $\ths'=\thf, \thr'=\min(\ths'+(\thr-\ths), n)$.

\one
\para{Distributed key generation.}
To setup the secret-shares of the \tc\ secret, we use the DKG protocol of~\cite{das2024distributed} with the following minor modifications. Each party starts by sharing the same secret independently using two weight thresholds $w$ and $w'$. The rest of the DKG protocol is identical to~\cite{das2024distributed}, except parties agree on two different aggregated PVSS transcript instead of one. Note that, these doubles the computation and communication cost of DKG.

\one
\para{Randomness generation.}
As described in~\Cref{alg:ramp}, parties reveal their VUF shares (\tc output shares) for the \fast upon prefinalizing a block, and for the \slow upon finalizing a block.
For both paths, the parties collect the VUF shares and are ready to execute the block as soon as the randomness (\tc output) is reconstructed from either path. 
As in~\cite{das2024distributed}, parties use the VUF output evaluated on the secret $s$ and the round (and epoch) number of the finalized block as the \tc\ output~\footnote{In the implementation of~\Cref{alg:ramp} with \jolteon as the \mbb protocol, it is safe to omit $\m$ in the evaluation function, since all honest parties always prefinalize the same block for each round as mentioned in section {\em consensus}.}.
Since the parties are weighted, similar to~\Cref{app:rand:overview}, we implement the weighted extension of \Cref{alg:ramp}.

\one
\para{Implementation Details.}
We implement \fastrand in Rust, atop the open-source \posrand~\cite{das2024distributed} implementation~\cite{aptos} on the \aptos blockchain.
We worked with \aptos to deploy our protocol on the \aptos blockchain.
For communication, we use the ${\tt Tokio}$ library~\cite{tokio}.
For cryptography, we use the ${\tt blstrs}$ library~\cite{blstrs}, which implements efficient finite field and elliptic curve arithmetic. Similar to \posrand~\cite{das2024distributed}, our implementation runs the share verification step of different parties in parallel and parallelizes the VUF derivation using multi-threading. 

\subsection{Evaluation Setup}
\label{app:rand:eval}

As of July 2024, the \aptos blockchain is run by $140$ validators, distributed $50$ cities across $22$ countries with the stake distributed described in~\cite{aptos-geo}. 
The $50$-th, $70$-th and $90$-th percentile (average) of round-trip latency between the blockchain validators is approximately $150$ms, $230$ms, and $400$ms, respectively.

Let $n$ denote the total stake before rounding, which is approximately $8.7\times 10^8$. 
The secrecy and reconstruction thresholds (in stakes) for the \slow are $\ths=0.5n$ and $\thr=0.66n$, respectively. 
The secrecy and reconstruction thresholds (in stakes) for the \fast are $\ths'=0.67n$ and $\thr'=0.83n$, respectively.
The total weight of the mainnet validators after rounding is $244$.
The weight threshold for the \slow is $w=143$, and that for the \fast is $w'=184$.

Most of the \aptos validators use the following recommended hardware specs~\cite{aptos-node}.
\begin{itemize}
    \item CPU: 32 cores, 2.8GHz or faster, AMD Milan EPYC or Intel Xeon Platinum.
    \item Memory: 64GB RAM.
    \item Storage: 2T SSD with at least 60K IOPS and 200MiB/s bandwidth.
    \item Network bandwidth: 1Gbps.
\end{itemize}

\para{Evaluation metrics.}
We measure the {\em randomness latency} as the duration required to generate randomness for each block, as in~\Cref{def:latency}.
It measures the duration from the moment the block is finalized by consensus to the when the randomness for that block becomes available. 
We report the average randomness latency (measured over a period of 12 hours).
We also measure and compare the setup overhead for \posrand~\cite{das2024distributed} and \fastrand, using microbenchmarks on machines of the same hardware specs as the \aptos mainnet.
As mentioned in~\Cref{app:rand:impl}, \fastrand requires the DKG to share the same secret in two different thresholds, resulting in approaximately twice the cost compared to \posrand~\cite{das2024distributed}.
We also measure the end-to-end DKG latency of \fastrand on the \aptos mainnet. 

\subsection{Evaluation Results.}
\input{appendix/eval_table}

\one
\para{Setup overhead.}
We report the computation costs and the transcript sizes of the setup in~\Cref{tab:dkg}. 
The end-to-end DKG latency for setting up \fastrand on \aptos mainnet is $16.8$ seconds. This includes $14.6$ seconds for transcript dealing, dissemination, and aggregation, and $2.2$ seconds for agreeing on the aggregated transcript. 
As observed in~\Cref{tab:dkg}, the computation overhead of DKG is a relatively small proportion of the end-to-end latency. As each party verifies the transcripts of other parties in parallel, the main bottleneck is the dissemination of the transcripts. 

Note that the setup overhead occurs only during the initial setup or when the set of parties changes (which happens every few hours or days), and does not affect the blockchain's end-to-end latency as the setup is performed asynchronously to blockchain transaction processing.
In contrast, the randomness latency increases the latency of every transaction. 
Thus, we believe that the significantly improved randomness latency at the cost of higher setup overhead is a reasonable trade off.

%% file: appendix/eval_table.tex
\begin{table}[t!]
    \centering
    \small
    \setlength\tabcolsep{3pt}
    \begin{tabular}{ccccc}
        \toprule
        \multirow{2}{*}{Scheme} & \multicolumn{3}{c}{DKG latency} & \multirow{2}{*}{\makecell{Transcript \\ size}}  \\
        \cmidrule{2-4}
        & deal & verify & aggregate \\
        \midrule
        \posrand~\cite{das2024distributed} & 190.2 & 171.8 & 1.6 & 80,021 bytes \\
        \fastrand & 377.4 & 351.6 & 3.1 & 160,041 bytes \\
        \bottomrule
    \end{tabular}
    \vspace{5mm}
    \caption{Setup overhead of \posrand~\cite{das2024distributed} and our \fastrand in microbenchmarks. 
    The latency unit is millisecond.}
    \label{tab:dkg}
\end{table}

%% file: appendix/discussion.tex
\section{Discussions}
\label{app:discussion}
As we discussed in~\Cref{sec:intro}, our results applies to many other threshold cryptosystems that are natively integrated into blockchains. Next, we provide one more specific example.

\one
\para{Threshold decryption for privacy.}
To counteract MEV and enhance privacy guarantees, numerous recent papers propose to use threshold decryption to hide transaction contents in a block until the block is finalized by the consensus protocol~\cite{reiter1994securely,yang2022sok,asayag2018fair,zhang2023f3b,momeni2022fairblock,kavousi2023blindperm,malkhi2022maximal}. 
These proposals work as follows. The blockchain validators start by secret sharing a decryption key using threshold secret sharing. Clients submit encrypted transactions to the blockchain, and the validators collectively decrypt the transactions in a block once the block is finalized. Our solutions (\Cref{alg:tight} and~\Cref{alg:ramp}) can also be used to improve the latency of such proposals that rely on non-interactive threshold decryption. However, there is a subtle issue that needs to be addressed first. Recall that in the definition of \longname, the \tc output is evaluated with respect to $(\rd,\m)$ where the message \m\ is finalized by \round \rd in \MBB. Therefore, each transaction must be encrypted for a specific round and finalized for that round in \MBB. We leave addressing this issue as an interesting open research direction.